\documentclass[12pt, onecolumn]{IEEEtran} 

\usepackage{url}
\usepackage{makecell}
\usepackage{amsfonts}
\usepackage{amssymb}
\usepackage{amsmath}
\usepackage{graphicx, colordvi, psfrag}
\usepackage{calc,pstricks, pgf, xcolor}
\usepackage{epsfig, cite}
\usepackage{bm} 
\usepackage{bbm} 
\usepackage{enumerate}

\newcommand{\beq}[1]{\begin{equation}\label{#1}}
\newcommand{\eeq}{\end{equation}}

\newcommand{\beqn}[1]{\begin{eqnarray}\label{#1}}
\newcommand{\eeqn}{\end{eqnarray}}

\newtheorem{thmbody}{Theorem}
\newenvironment{thm}{
\begin{thmbody}
	}{
	\end{thmbody} 
	}
\newtheorem{dfnbody}{Definition}

\newtheorem{corbody}{Corollary}

\newtheorem{lemmabody}{Lemma}
\newenvironment{lemma}{
\begin{lemmabody}
	}{
	\end{lemmabody} 
	}
\newtheorem{propbody}{Proposition}

\newenvironment{proof}{
	{\it Proof:}
	}{
 $\Box$
	}

\usepackage{dsfont}
\begin{document}
\title{A Generalization of the DMC}

\author{
\IEEEauthorblockN{Sergey Tridenski and Anelia Somekh-Baruch}\\
\IEEEauthorblockA{Faculty of Engineering\\Bar-Ilan University\\ Ramat-Gan, Israel\\
Email: tridens@biu.ac.il, somekha@biu.ac.il}
\thanks{
The material in this paper was partially presented in ITW2020 \cite{TridenskiSomekhBaruch20}.}
\thanks{
This work was supported by the Israel Science Foundation (ISF) under grant \#631/17.}
}



\maketitle
\begin{abstract}
We consider a generalization of the discrete memoryless channel, in which the channel probability distribution is
replaced by a uniform distribution over clouds of channel output sequences.
For a random ensemble of such channels, we derive an achievable error exponent, as well as its converse together with the optimal correct-decoding exponent,
all as functions of information rate.
As a corollary of these results, we obtain the channel ensemble capacity.
\end{abstract}



%
%
%
%


\bigskip

\section{Introduction}\label{Int}

\bigskip

We consider the basic information-theoretic scenario of point-to-point communication.
The standard go-to model for such a scenario is the discrete memoryless channel (DMC).
With this model, the communication performance is characterized by the
channel capacity, 
surrounded by 
the error and the correct-decoding exponents --- as functions of the 
information rate.
In order to be evaluated
by these 
quantities,
the communication is usually done with a codebook of blocks of $n$ channel input symbols,
conveying $2^{nR}$ equiprobable
messages, where $R$ is the rate in bits. 

In this paper we
slightly deviate
from the standard DMC model.
In 
our set-up,
the DMC itself reappears as a limiting case.
Consider first the following communication scheme.
Let $K$ be some positive real parameter 
in addition to the rate $R$,
and suppose that there has been an exponentially large number $2^{n(R\,+\,K)}$ of block transmissions through a DMC.
Each transmitted block is a codeword of length $n$,
chosen each time with uniform probability from the same codebook of size $2^{nR}$.
This corresponds to a 
significant
amount of transmitted data of $nR \cdot 2^{n(R\,+\,K)}$ bits.
By the end of these transmissions, each of the $2^{nR}$ codewords has been used approximately $2^{nK}$ times, resulting in $2^{nK}$ 
not necessarily distinct channel outcomes, forming an unordered ``cloud''. The parameter $K$ therefore represents an exponential size of the cloud of channel output vectors generated by a single codeword.
Suppose that in the end of the $2^{n(R\,+\,K)}$ transmissions
the outcome clouds of all the codewords are revealed noiselessly to the decoder.
For small $K$, when most of the output vectors in the clouds are distinct,
this ``revelation'' would be approximately equivalent to a noiseless transmission of the same $nR \cdot 2^{n(R\,+\,K)}$ bits of data.
For higher $K$, however, the description of the clouds will require an exponentially smaller number of noiseless bits
compared to $nR \cdot 2^{n(R\,+\,K)}$.

Note, that given the $2^{n(R\,+\,K)}$ received channel-output blocks, 
with time indices $j = 1, ..., 2^{n(R\,+\,K)}$, and the knowledge of the clouds, the optimal decoder for any given output block $j$
(in the sense that it minimizes the probability of error for the block $j$) 
chooses the codeword with the maximal number of replicas of this block
in its cloud.
This decoder is optimal regardless
of the message probabilities or the transition probabilities of the DMC, 
that created the clouds.
Moreover,
the same decoder, which relies on clouds and is oblivious of the channel transition probabilities,
remains optimal whether or not the channel is memoryless
within each
block. 

Given the clouds,
the receiver sees effectively 
a different channel ---
one which chooses its output vector with uniform probability from the cloud of the sent codeword.
This channel can be described by a model, different from DMC.
In this model, we assume that the messages are equiprobable and each cloud contains exactly $2^{nK}$ vectors.
The clouds are generated randomly i.i.d. with a channel-generating distribution,
independently for each codeword in a codebook.
This is similar to constant composition clouds used for superposition coding \cite{SomekhBaruch15} through a noiseless channel.
The capacity and the relevant probability exponents of this scheme can be given in the average sense, for the ensemble of random channels.
As the exponential size of the clouds $K$ tends to infinity,
the random channel ensemble converges to a single channel with the transition probabilities of the channel generating distribution,
which is a DMC in our case \cite{GallagerBook}, \cite{Gallager73}, \cite{Arimoto73}.

This paper is organized as follows.
In Section~\ref{Mod} we start introducing our notation and define the channel model. In Section~\ref{Ach} we derive an achievable error exponent for the random channel ensemble.
In Sections~\ref{Con},~\ref{AltRep} we provide 
converse results. We derive
an upper bound on the optimal error exponent (in Section~\ref{Con})
and
the optimal correct-decoding exponent (in Sections~\ref{Con},~\ref{CorDec}) of the random channel ensemble.
In Section~\ref{Cap} we obtain the channel ensemble capacity, as a corollary of the previous sections.


\bigskip

\section{Channel model}\label{Mod}

\bigskip

Let $x \in {\cal X}$ and $y \in {\cal Y}$ be letters from finite channel-input and -output alphabets, respectively.
Let $W(y \, | \, x)$ denote transition probabilities of a channel-{\em generating} distribution.
The channel is generated for a given 
codebook
of blocks of a length $n$ of letters from ${\cal X}$.
Let ${\cal C\mathstrut}_{\!n}$ be such a codebook,
consisting of $\lfloor e^{nR}\rfloor$ codewords
${\bf x\mathstrut}_{m} \in {\cal X\mathstrut}^{n}$, $m = 1, 2, ..., \lfloor e^{nR}\rfloor$,
where $R$ is a positive real number representing a communication rate.

Given this codebook and another positive real number $K$, a channel {\em instance} is generated with the distribution $W$, as follows.
For each one of the 
$\lfloor e^{nR}\rfloor$
messages $m$
an exponentially large number $\lfloor e^{nK}\rfloor$ of sequences ${\bf y} \in {\cal Y\mathstrut}^{n}$
is generated randomly given the corresponding codeword ${\bf x\mathstrut}_{m}$, where each sequence ${\bf y}$ is generated independently of others. Each letter ${y\mathstrut}_{i}$, $i = 1, 2, ..., n$, of each such sequence ${\bf y}$ is generated i.i.d. according to $W$ given the corresponding letter ${x\mathstrut}_{m i}$ of ${\bf x\mathstrut}_{m}$. In this way, the set of clouds of ${\bf y}$'s of size $\lfloor e^{nK}\rfloor$ for each $m$
forms one channel instance.

We assume that the messages $m = 1, 2, ..., \lfloor e^{nR}\rfloor$,
represented by the codewords ${\bf x\mathstrut}_{m}$, are equiprobable.
Given that a particular 
message
is sent through the channel,
the stochastic channel {\em action} now amounts to choosing exactly one of all the not necessarily unique $\lfloor e^{nK}\rfloor$ vectors ${\bf y}$,
corresponding to the sent 
message,
with the uniform probability $1/\lfloor e^{nK}\rfloor$.
We assume that the decoder, receiving the channel output vector $\widehat{\bf y}$, knows not only the codebook, but also the channel instance, i.e., all the $\lfloor e^{nR}\rfloor$ clouds comprising the corresponding ${\bf y}$'s.

A cloud can have more than one replica of the received vector $\widehat{\bf y}$.
The maximum-likelihood (ML) decoder makes an error with non-vanishing probability $\geq \frac{1}{2}$,
if there exists an incorrect message 
with the same or a higher number of vectors $\widehat{\bf y}$ in its cloud,
comparing to
the sent message 
itself. Otherwise there is no error.


\bigskip

\section{Achievable error exponent}\label{Ach}

\bigskip

Suppose the codebook is generated i.i.d. according to a distribution over ${\cal X}$ with probabilities $P(x)$.
Let $\,\overline{\!{P\mathstrut}}{\mathstrut}_{\!e}^{\,(n)}$ denote the average error probability 
of the maximum-likelihood decoder, averaged over all possible messages, codebooks\footnote{In the next section we prove our converse result, which does not assume random codes, but is valid for the best sequence of block codes. Together with the random-coding result of Section~\ref{Ach} this will give the channel ensemble capacity.}, and channel instances.
Let $T(y)V(x \, | \, y)$ denote probabilities in an auxiliary distribution over ${\cal X}\times {\cal Y}$ and let us define:
\begin{align}
{A\mathstrut}_{TV}^{P}
\;\; &
\triangleq \;\; D(V \, \| \, P \, | \, T), \;\;\;\;\;\;
\;\;\;\;\;\;\;\;\;\;\;
{H\mathstrut}_{\!T}
\;\; 
\triangleq \;\; {\mathbb{E}\mathstrut}_{T}\big[\!-\!\log T(Y)\big], 
\nonumber \\
{B\mathstrut}_{TV}^{W}
\;\; &
\triangleq \;\; {\mathbb{E}\mathstrut}_{TV}\big[\!-\!\log W(Y\,|\,X)\big],
\label{eqDefs} \\
{E\mathstrut}_{\!e}(P, R, K) \;\; & \triangleq \;\;
\min_{\substack{\\ TV}}\;
\Big\{
D(TV \, \| \, PW) + \big|{A\mathstrut}_{TV}^{P}-R + \big|{B\mathstrut}_{TV}^{W} - K \big|^{+}\big|^{+}\Big\},
\label{eqOneDef}
\end{align}
where $D(V\,\|\,P\,|\,T) = \sum_{x, \,y}T(y)V(x\,|\,y)\log\frac{V(x\,|\,y)}{P(x)}$ is the Kullback-Leibler divergence averaged over $T$, the expectation ${\mathbb{E}\mathstrut}_{TV}$ 
is with respect to the joint distribution $TV$, and $\,|\,t\,|^{+}\triangleq \max\,\{0, t\}$. All the logarithms 
here and 
below
are to base $e$. In what follows, we usually suppress the superscripts in ${A\mathstrut}_{TV}$ and ${B\mathstrut}_{TV}$.
Then we can show the following

\bigskip

\begin{thm}[Random-coding error exponent]\label{thmAch}
\begin{equation} \label{eqPerrorBoth}
\lim_{n\,\rightarrow\,\infty} \,-\tfrac{1}{n}\log 
\,\overline{\!{P\mathstrut}}{\mathstrut}_{\!e}^{\,(n)}
\;\; = \;\;
{E\mathstrut}_{\!e}(P, R, K),
\end{equation}
{\em where ${E\mathstrut}_{\!e}(P, R, K)$ is defined in (\ref{eqOneDef}).}
\end{thm}

\bigskip

We 
prove this theorem by separately deriving matching lower and upper bounds.
For the lower bound,
for $\epsilon \in \mathbb{R}$ let us further define
\begin{align}
{E\mathstrut}_{\!e}(P, R, K, \epsilon) \; & \triangleq \;\,
\;\;\;\;\;\;\;\,
\min\;
\;\;\;\;\;\;
\big\{{E\mathstrut}_{1}(P, R, K, \epsilon), \, {E\mathstrut}_{2}(P, R, K, \epsilon)\big\},
\label{eqMinofTwo} \\
{E\mathstrut}_{1}(P, R, K, \epsilon) \; & \triangleq
\min_{\substack{\\ TV\widetilde{V}:\\
{B\mathstrut}_{T\widetilde{V}}  \; \leq \;
{B\mathstrut}_{TV} \, + \, \epsilon \; \leq \; K
}}
F(TV, \, T\widetilde{V}),
\label{eqMany} \\
{E\mathstrut}_{2}(P, R, K, \epsilon) \; & \triangleq
\;\;\;\;\,
\min_{\substack{\\ TV\widetilde{V}:\\
{B\mathstrut}_{TV} \, + \, \epsilon \; \geq \; K}}
\;\;\;\;\,
F(TV, \, T\widetilde{V}),
\label{eqSingle} \\
F(TV, \, T\widetilde{V}) \; & \triangleq \,
D(TV \, \| \, PW) + \big|{A\mathstrut}_{T\widetilde{V}} - R \, +
\,\big|{B\mathstrut}_{\!T\widetilde{V}} - K \big|^{+}\big|^{+}
,
\label{eqF}
\end{align}
where
${A\mathstrut}_{T\widetilde{V}}$ and ${B\mathstrut}_{TV}$ are defined by (\ref{eqDefs}). 
Our lower bound is given by Lemma~\ref{thm1}, together with Lemmas~\ref{lemmaZeroEps} and~\ref{lemma1} below.

\bigskip

\begin{lemma}[Lower bound] \label{thm1}
\begin{equation} \label{eqPerror}
\liminf_{n\,\rightarrow\,\infty} \,-\tfrac{1}{n}\log 
\,\overline{\!{P\mathstrut}}{\mathstrut}_{\!e}^{\,(n)}
\;\; \geq \;\;
\lim_{\;\;\epsilon \, \searrow \, {0\mathstrut}^{+}}
{E\mathstrut}_{\!e}(P, R, K, \epsilon),
\end{equation}
{\em where ${E\mathstrut}_{\!e}(P, R, K, \epsilon)$ is defined in (\ref{eqMinofTwo})-(\ref{eqF}).}
\end{lemma}

\bigskip

In the proof of Lemma~\ref{thm1} we use the following auxiliary lemma:

\bigskip

\begin{lemma}[Super-exponential bounds] \label{lemmaA}

{\em Let ${Z\mathstrut}_{\!i}\,$, $i = 1$, $2$, ... , be i.i.d. \text{Bernoulli}($e^{-nB}$) random variables. Then for any $\delta > 0$}
\begin{align}
\Pr \Bigg\{\sum_{i \, = \, 1}^{\lfloor e^{nK}\rfloor}{Z\mathstrut}_{\!i}
\;
\geq
\;\; e^{n\big[|K\,-\,B|^{+} +\,\delta\big]}\Bigg\}
\; & \leq \;
\exp\bigg\{ - \, e^{n\big[|K\,-\,B|^{+} +\,\delta \,+\,o(1)\big]}\bigg\},
\label{eqIncor} \\
\Pr \Bigg\{\sum_{i \, = \, 1}^{\lfloor e^{nK}\rfloor}{Z\mathstrut}_{\!i}
\;
\leq
\;\; e^{n\big[K\,-\,B \, - \,\delta\big]}\Bigg\}
\; & \leq \;
\exp\bigg\{ - \, e^{n\big[K\,-\,B \,+\,o(1)\big]}\bigg\},
\label{eqCor}
\end{align}
{\em where $o(1)$ is a function of $(\delta, K)$ that satisfies $o(1) \rightarrow 0$ as $n \rightarrow \infty$.}
\end{lemma}
\begin{proof}
The result follows straightforwardly from
Markov's inequality for the random variable $e^{\sum_{i}{Z\mathstrut}_{\!i}}$,
resp. $e^{-\sum_{i}{Z\mathstrut}_{\!i}}$, and the inequality $1 + t \leq e^{t}$.
\end{proof}

\bigskip

{\em Proof of Lemma~\ref{thm1}:}
We will use $\epsilon > 0$ to establish 
(\ref{eqPerror}).

Let ${\bf x}$ be the sent codeword and $\widehat{\bf y}$ be the received vector.
The cloud of ${\bf x}$ can 
contain
more than one vector $\widehat{\bf y}$.
The maximum-likelihood decoder makes an error with non-vanishing probability $\geq \frac{1}{2}$,
if there exists an incorrect codeword 
(not necessarily distinct from ${\bf x}$, but representing a different message and having therefore an independently generated cloud) with the same or a higher number of vectors $\widehat{\bf y}$ in its cloud,
compared to
the sent codeword 
itself. Otherwise there is no error.

Consider an event 
where
${\bf x}$ and $\widehat{\bf y}$
have a joint empirical distribution (type with denominator $n$) $TV$, i.e., $TV \in {{\cal P}\mathstrut}_{\!n}({\cal Y}\times {\cal X})$,
where $T$ is a distribution on ${\cal Y}$ and $V$ is a conditional distribution on ${\cal X}$ given a letter from ${\cal Y}$.
The 
exponent of the probability of 
this event (probability of type class in \cite{CoverThomas}) is given by 
\begin{displaymath}
D(TV \, \| \, PW) + o(1),
\end{displaymath}
where the term $o(1)$ vanishes uniformly w.r.t. $TV$, as $n\rightarrow \infty$.

Consider now the competing codewords. The exponent of the probability of an event, that $\widehat{\bf y}$ appears 
somewhere in the 
clouds corresponding to 
the incorrect codewords,
is given by:
\begin{equation} \label{eqAtLeastOne}
\min_{\substack{\\\widetilde{V}: \;\; T\widetilde{V}\,\in\,{{\cal P}\mathstrut}_{\!n}({\cal Y}\,\times\, {\cal X})} }\;
\Big\{
\big|{A\mathstrut}_{T\widetilde{V}} - R + \big|{B\mathstrut}_{T\widetilde{V}} - K \big|^{+}\big|^{+}\Big\} + o(1),
\end{equation}
where $o(1)$ is uniform w.r.t. $T$.
To see this,
consider a possibly different (from $V$) conditional type $\widetilde{V}$ of some $\widetilde{\bf x}$
w.r.t. $\widehat{\bf y}$.
The exponent of the probability of an event, that a certain incorrect codeword belongs to the conditional type 
$\widetilde{V}$ given $\widehat{\bf y}$,
is given by 
\begin{equation} \label{eqExponentA}
{A\mathstrut}_{T\widetilde{V}} + o(1),
\end{equation}
where $o(1)$ is uniform w.r.t. $T\widetilde{V}$. The exponent of the probability of an event, that a certain ${\bf y}$ in the cloud of $\widetilde{\bf x}$
equals $\widehat{\bf y}$, is given by ${B\mathstrut}_{T\widetilde{V}}\,$.
The exponent of the probability of an event, that in the cloud of $\widetilde{\bf x}$ of the type $\widetilde{V}$ 
the vector $\widehat{\bf y}$ appears at least once,
is given by
\begin{equation} \label{eqExponentK}
\big|{B\mathstrut}_{T\widetilde{V}} - K\big|^{+} + o(1),
\end{equation}
where the term $o(1)$, 
vanishing as $n \rightarrow \infty$,
depends on
$K$.
In particular, as a {\em lower} bound on the exponent, (\ref{eqExponentK}) follows trivially without $o(1)$ from the union bound on the probability.
Whereas
to confirm (\ref{eqExponentK}) as an {\em upper} bound on the exponent,
denoting $e^{-n{B\mathstrut}_{T\widetilde{V}}}\triangleq \alpha$ and $\lfloor e^{nK}\rfloor \triangleq M \equiv \beta^{-1}$,
we can write similarly to \cite[Eq.~14]{Gallager73}:
\begin{align}
& \Pr\big\{\text{$\widehat{\bf y}$ is in the cloud of $\widetilde{\bf x}$}\big\} \;\;
= \;\; 1 - (1 - \alpha)^{M}
\nonumber \\
& \equiv \;\; \alpha \sum_{j\,=\,0}^{M - 1} {(1 - \alpha)\mathstrut}^{\,j}
\;\; \geq \;\; \min \,\{\alpha, \, \beta\} \sum_{j\,=\,0}^{M - 1} {(1 - \beta)\mathstrut}^{\,j}
\nonumber \\
&
\;\;\;\;\;\;\;\;\;\;\;\;\;\;\;\;\;\;\;\;\;\;\;\;\;\;\;\;\;\;\,\,
\equiv \;\; \min \,\{\alpha \beta^{-1}, \,1\} \,\big[1 - \underbrace{(1 - \beta)^{M}\!}_{\leq \; 1/e}\big]
\nonumber \\
&
\;\;\;\;\;\;\;\;\;\;\;\;\;\;\;\;\;\;\;\;\;\;\;\;\;\;\;\;\;\;\,\,
= \;\; e^{-n(|{B\mathstrut}_{T\widetilde{V}} \,-\, K|^{+} \, + \, o(1))},
\nonumber
\end{align}
where $o(1)$ is a function of $K$.
Adding together (\ref{eqExponentA}) and (\ref{eqExponentK}),
we obtain
the exponent of the probability of an event, that a certain incorrect codeword is of the conditional type $\widetilde{V}$ w.r.t. $\widehat{\bf y}$,
{\em and} 
$\widehat{\bf y}$ appears at least once
in its cloud:
\begin{equation} \label{eqExponentS}
{A\mathstrut}_{T\widetilde{V}} + \big|{B\mathstrut}_{T\widetilde{V}} - K \,\big|^{+} + o(1),
\end{equation}
where $o(1)$ is uniform w.r.t. $T\widetilde{V}$.
Finally, the exponent of the probability of an event, that in the codebook there exists at least one incorrect codeword of the conditional type $\widetilde{V}$ w.r.t. $\widehat{\bf y}$,
and $\widehat{\bf y}$ appears at least once
in its cloud, is given by
\begin{equation} \label{eqExponentF}
\big|{A\mathstrut}_{T\widetilde{V}} + \big|{B\mathstrut}_{T\widetilde{V}} - K \,\big|^{+}- R \,\big|^{+} + o(1),
\end{equation}
where $o(1)\rightarrow 0$ uniformly w.r.t. $T\widetilde{V}$ as $n \rightarrow \infty$ and may depend on $K$ and $R$,
which yields (\ref{eqAtLeastOne}).

Suppose that $K - {B\mathstrut}_{TV} \leq \epsilon$. In this case 
the exponent of the conditional probability
of error given
that the received vector and the sent codeword
belong to the joint type
$TV$ can be lower bounded
by (\ref{eqAtLeastOne}), and the exponent of the (unconditional) probability of error due to all such cases is lower-bounded by
\begin{equation} \label{eqAtLeastOneOverall}
\min_{\substack{\\ TV\widetilde{V}: \\
TV, \, T\widetilde{V}\;\in\;{{\cal P}\mathstrut}_{\!n}({\cal Y}\,\times\, {\cal X}),
\\
K \, - \, {B\mathstrut}_{TV} \; \leq \; \epsilon}}
\!\!
\Big\{
D(TV \, \| \, PW) + \big|{A\mathstrut}_{T\widetilde{V}}-R + \big|{B\mathstrut}_{T\widetilde{V}} - K \big|^{+}\big|^{+}\Big\}
+ o(1).
\end{equation}

Consider now the opposite case when $K - {B\mathstrut}_{TV} \geq \epsilon$.
For this case, recall that the exponent of the probability of an event, that 
there
exists at least one incorrect codeword of the conditional type $\widetilde{V}$ w.r.t. $\widehat{\bf y}$,
is given by $\big|{A\mathstrut}_{T\widetilde{V}} - R\,\big|^{+} + o(1)$.
Suppose now that the conditional type $\widetilde{V}$ is such that $K - {B\mathstrut}_{T\widetilde{V}} \leq K - {B\mathstrut}_{TV} - \epsilon$.
For this case
we use
Lemma~\ref{lemmaA},
with $\delta = \epsilon/2$.
Using (\ref{eqCor}) for the correct cloud and (\ref{eqIncor}) for the competing clouds,
the probability of the event that the cloud of an incorrect codeword of the type $\widetilde{V}$ has
at least as many occurrences of the vector $\widehat{\bf y}$, compared to the correct codeword of the type $V$,
can be upper-bounded uniformly by
\begin{equation} \label{eqSuperExp}
\underbrace{\exp\big\{\! - e^{n(\epsilon + o(1))}\big\}}_{\text{the correct cloud}} \, + \,
\underbrace{\exp\big\{nR -e^{n(\epsilon/2 + o(1))}\big\}}_{\text{the competing clouds}}.
\end{equation}
That is, it tends to zero {\em super}-exponentially fast with $n$.
The remaining types $\widetilde{V}$ with $K - {B\mathstrut}_{T\widetilde{V}} \geq K - {B\mathstrut}_{TV} - \epsilon$
allow us to write a lower bound on the exponent of the (unconditional) probability of error due to all the cases $K - {B\mathstrut}_{TV} \geq \epsilon$,
as
\begin{equation} \label{eqManyOneOverall}
\min_{\substack{\\ TV\widetilde{V}:\\
TV, \, T\widetilde{V}\;\in\;{{\cal P}\mathstrut}_{\!n}({\cal Y}\,\times\, {\cal X}),
\\
K \, - \, {B\mathstrut}_{T\widetilde{V}} \; \geq \; 
K \, - \, {B\mathstrut}_{TV} \, - \, \epsilon \; \geq \; 0}}
\!\!
\Big\{
D(TV \, \| \, PW) + \big|{A\mathstrut}_{T\widetilde{V}} - R \, \big|^{+}\Big\} + o(1).
\end{equation}
Together, (\ref{eqAtLeastOneOverall}) and (\ref{eqManyOneOverall}) cover all cases and the minimum between the two gives the lower bound on the error exponent.

Observe that the objective function of (\ref{eqAtLeastOneOverall}) can be used also in (\ref{eqManyOneOverall}), because
in (\ref{eqManyOneOverall}) the set over which the minimization is performed
satisfies
${B\mathstrut}_{T\widetilde{V}} \leq K$.
Furthermore, for the lower bound, we can simply extend the minimization set in (\ref{eqAtLeastOneOverall}) and (\ref{eqManyOneOverall}) from types to arbitrary distributions
$TV$ and $T\widetilde{V}$.
Therefore,
omitting $o(1)$, in the limit of a large $n$ we can replace the minimum of the bounds (\ref{eqAtLeastOneOverall}) and (\ref{eqManyOneOverall}) with
(\ref{eqMinofTwo}). $\square$

\bigskip

To complete the lower bound given by Lemma~\ref{thm1}, we establish the next two lemmas.

\bigskip

\begin{lemma}[Epsilon equals zero] \label{lemmaZeroEps}
{\em The expression defined in (\ref{eqMinofTwo})-(\ref{eqF}) satisfies}
\begin{equation} \label{eqConvergeFromBelow}
\lim_{\;\;\epsilon \, \searrow \, {0\mathstrut}^{+}}
{E\mathstrut}_{\!e}(P, R, K, \epsilon)
\;\; = \;\; {E\mathstrut}_{\!e}(P, R, K, 0).
\end{equation}
\end{lemma}
\begin{proof}
Observe first that both (\ref{eqMany}) and (\ref{eqSingle}) are convex ($\cup$) functions of $\epsilon \in \mathbb{R}$.
This can be verified directly by the definition of convexity, using the property that $F(TV, \,T\widetilde{V})$ is convex ($\cup$)
and ${B\mathstrut}_{TV}$ is linear in the pair $(TV, \,T\widetilde{V})$.
Furthermore, by continuity of $F(TV, \,T\widetilde{V})$ and ${B\mathstrut}_{TV}$, it follows that (\ref{eqMany}) and (\ref{eqSingle}) are {\em lower semi-continuous}
functions of $\epsilon \in \mathbb{R}$.
Observe next from (\ref{eqMany}) and (\ref{eqSingle}) that at least one of them is necessarily finite at $\epsilon = 0$, i.e., ${E\mathstrut}_{\!e}(P, R, K, 0) < +\infty$.
Suppose that ${E\mathstrut}_{2}(P, R, K, 0) \leq {E\mathstrut}_{1}(P, R, K, 0)$.
Then ${E\mathstrut}_{2}(P, R, K, \epsilon)$ is finite for $\epsilon \geq 0$ and by the lower semi-continuity of the convex function $\lim_{\,\epsilon \,\searrow\, 0}{E\mathstrut}_{2}(P, R, K, \epsilon) = {E\mathstrut}_{\!e}(P, R, K, 0)$. Then we also obtain (\ref{eqConvergeFromBelow}).
Consider the opposite case ${E\mathstrut}_{1}(P, R, K, 0) < {E\mathstrut}_{2}(P, R, K, 0)$.
Then (\ref{eqMany}) at $\epsilon = 0$ is a minimization of a continuous function of $TV\widetilde{V}$ over a closed non-empty set.
Let $\hat{T}\hat{V}$ be the distribution $TV$, achieving the minimum in (\ref{eqMany}) at $\epsilon = 0$.
Then necessarily $K > {B\mathstrut}_{\hat{T}\hat{V}}$ (otherwise with $K = {B\mathstrut}_{\hat{T}\hat{V}}$ there has to be
${E\mathstrut}_{1}(P, R, K, 0) \geq {E\mathstrut}_{2}(P, R, K, 0)$).
Then ${E\mathstrut}_{1}(P, R, K, \epsilon)$ is finite for $K - {B\mathstrut}_{\hat{T}\hat{V}} > \epsilon \geq 0$
and by the lower semi-continuity of the convex function $\lim_{\,\epsilon \,\searrow\, 0}{E\mathstrut}_{1}(P, R, K, \epsilon) = {E\mathstrut}_{\!e}(P, R, K, 0)$.
Then again we obtain (\ref{eqConvergeFromBelow}).
\end{proof}

\bigskip

{\em Remark:}

Alternatively, the achievable result of Lemmas~\ref{thm1},~\ref{lemmaZeroEps}
can also be obtained using a suboptimal decoder,
which does not count the exact number of replicas of the received vector $\widehat{\bf y}$ in each cloud:
\begin{equation} \label{eqAlternativeDec}
\widehat{m} \; = \; \underset{\substack{m: \;\text{at least one $\widehat{\bf y}$ in the cloud}}}{\arg\max} \,
\max
\big\{\!-\!K, \, -{B\mathstrut}_{TV}\big\},
\end{equation}
where $TV$ denotes the joint type of the pair $(\widehat{\bf y}, {\bf x\mathstrut}_{m})$.

\bigskip

\begin{lemma}[Identity] \label{lemma1}
\begin{align}
& {E\mathstrut}_{\!e}(P, R, K, 0) \; \equiv \; {E\mathstrut}_{\!e}(P, R, K),
\label{eqOneMin}
\end{align}
{\em where the LHS and the RHS are defined by (\ref{eqMinofTwo})-(\ref{eqF}) and (\ref{eqOneDef}), respectively.}
\end{lemma}
\begin{proof}
For $\epsilon = 0$, we can conveniently rewrite the
minimum (\ref{eqMinofTwo}) 
between
(\ref{eqMany}) and (\ref{eqSingle}) 
in the following unified manner:
\begin{align}
& {E\mathstrut}_{\!e}(P, R, K, 0) \; =
\min_{\substack{\\ TV\widetilde{V}:\\
K \, - \, {B\mathstrut}_{TV} \; \leq \; {|K \, - \, {B\mathstrut}_{T\widetilde{V}}|\mathstrut}^{+}}}
\Big\{\underbrace{-{H\mathstrut}_{\!T} + {A\mathstrut}_{TV} + {B\mathstrut}_{TV}}_{= \; D(TV\,\|\,PW)}
\nonumber \\
&
\;\;\;\;\;\;\;\;\;\;\;\;\;
+ \big|{A\mathstrut}_{T\widetilde{V}} +
\underbrace{{B\mathstrut}_{T\widetilde{V}}
+ \big|K - {B\mathstrut}_{T\widetilde{V}}\big|^{+} - K}_{= \; {|{B\mathstrut}_{T\widetilde{V}} \, - \, K|\mathstrut}^{+}} - \, R \,\big|^{+}\Big\},
\label{eqSingleMin}
\end{align}
where in the objective function we used also the property that $|\,t\,|^{+} = t + |-t\,|^{+}$.
Now it is convenient to 
verify,
that in (\ref{eqSingleMin}) the conditional distribution $\widetilde{V}$ without loss of optimality can be replaced with $V$.
To this end suppose that some joint distributions $TV$ and $T\widetilde{V}$
satisfy the condition under the minimum of (\ref{eqSingleMin}).

If ${A\mathstrut}_{TV} + {B\mathstrut}_{TV} \leq {A\mathstrut}_{T\widetilde{V}} + {B\mathstrut}_{T\widetilde{V}}$,
then, since also $\big|K - {B\mathstrut}_{TV}\big|^{+} \leq \,\big|K - {B\mathstrut}_{T\widetilde{V}}\big|^{+}$,
we cannot increase the objective function of (\ref{eqSingleMin})
by using $TV$ in place of $T\widetilde{V}$.

If ${A\mathstrut}_{TV} + {B\mathstrut}_{TV} \geq {A\mathstrut}_{T\widetilde{V}} + {B\mathstrut}_{T\widetilde{V}}$,
then we cannot increase the objective function of (\ref{eqSingleMin})
by using $T\widetilde{V}$ in place of $TV$.

It follows that 
(\ref{eqOneDef})
is a lower bound on minimum (\ref{eqSingleMin}).
Finally, since 
(\ref{eqOneDef})
is also an upper bound on (\ref{eqSingleMin}), we conclude that there is equality between 
(\ref{eqOneDef})
and (\ref{eqSingleMin}).
\end{proof}

\bigskip

Combining (\ref{eqOneMin}), (\ref{eqConvergeFromBelow}), and (\ref{eqPerror}), we have that the RHS of (\ref{eqPerrorBoth}) is a lower bound.
It remains to show that it is also an upper bound.

\bigskip

\begin{lemma}[Upper bound] \label{lemmaUp}
\begin{equation} \label{eqPerrorUp}
\limsup_{n\,\rightarrow\,\infty} \,-\tfrac{1}{n}\log 
\,\overline{\!{P\mathstrut}}{\mathstrut}_{\!e}^{\,(n)}
\;\; \leq \;\;
{E\mathstrut}_{\!e}(P, R, K),
\end{equation}
{\em where ${E\mathstrut}_{\!e}(P, R, K)$ is defined in (\ref{eqOneDef}).}
\end{lemma}

\bigskip

In the proof of Lemma~\ref{lemmaUp} we use the following auxiliary lemma:

\bigskip

\begin{lemma}[Two competing clouds] \label{lemmaB}

{\em Let $N \sim {\rm B}(M, \alpha)$ and ${N\mathstrut}_{\!1} \sim {\rm B}(M-1, \alpha)$ be two statistically independent binomial random variables
with the parameters $M \geq 2$ and $\alpha \in (0, q\,]\cup \{1\}$, where $q \in (0, 1)$ is a constant. Then}
\begin{equation} \label{eqConstBound}
\Pr\left\{N \, \geq \, {N\mathstrut}_{\!1} + 1 \, | \, N \, \geq \, 1\right\}
\;\; \geq \;\; \frac{1}{2}\big[1 - 1/\sqrt{2\pi}\,\big] + {o\mathstrut}_{M}(1),
\end{equation}
{\em where ${o\mathstrut}_{M}(1)$ depends on $q$ and as $M \rightarrow +\infty\,$ satisfies $\,{o\mathstrut}_{M}(1)\rightarrow 0$.}
\end{lemma}
The proof is given in the Appendix. In the above Lemma, $N$ and ${N\mathstrut}_{\!1} + 1$ can describe the random numbers of replicas of $\widehat{\bf y}$ in an incorrect cloud and the correct cloud, respectively.

\bigskip

{\em Proof of Lemma~\ref{lemmaUp}:} For the upper bound it is enough to consider the exponent of the probability of the event that
the transmitted and the received blocks
${\bf x}$ and $\widehat{\bf y}$
have a joint type $TV$, while
in the codebook there exists at least one incorrect codeword of the {\em same} conditional type $V$ w.r.t. $\widehat{\bf y}$,
and $\widehat{\bf y}$ appears at least once
in its cloud. As in the proof of Lemma~\ref{thm1}, this exponent is given by
\begin{equation} \label{eqJustTV}
D(TV \, \| \, PW) + \big|{A\mathstrut}_{TV}-R + \big|{B\mathstrut}_{TV} - K \big|^{+}\big|^{+}
+ o(1).
\end{equation}
The additional exponent of the conditional probability of error given this event is $o(1)$, as follows immediately by Lemma~\ref{lemmaB}, used with $M = \lfloor e^{nK}\rfloor$
and $\alpha = e^{-n{B\mathstrut}_{TV}}$ with $q = \max_{\,W(y\,|\,x) \, < \, 1}W(y\,|\,x) > 0$, or $\alpha = 1$.
In the limit of a large $n$ , we can omit $o(1)$ and by continuity minimize (\ref{eqJustTV}) over all distributions $TV$, to obtain the RHS of (\ref{eqPerrorUp}).
$\square$

\bigskip

This completes the proof of Theorem~\ref{thmAch}. An alternative 
representation of
the error exponent of Theorem~\ref{thmAch} is given by

\bigskip

\begin{lemma}[Dual form] \label{lemma11}
\begin{equation} \label{eqExplicit1}
{E\mathstrut}_{\!e}(P, R, K) \; = \;
\sup_{\substack{\\0\,\leq\,\eta \, \leq \, \rho \, \leq \, 1}}
\Big\{{E\mathstrut}_{0}(\rho, \eta, P) - \rho R - \eta K\Big\},
\end{equation}
{\em where ${E\mathstrut}_{\!e}(P, R, K)$ is defined in (\ref{eqOneDef}) and}
\begin{equation} \label{eqE0Def}
{E\mathstrut}_{0}(\rho, \eta, P) \; \triangleq \;
-\log \sum_{y}\Big[\sum_{x}P(x)W^{\frac{1\,+\,\eta}{1\,+\,\rho}}(y\,|\,x)\Big]^{1\,+\,\rho}.
\end{equation}
\end{lemma}

\bigskip

\begin{proof}
Observe first that the minimum (\ref{eqOneDef}) can be lower-bounded as
\begin{align}
&
\min_{\substack{\\ TV}}\;
\Big\{D(TV \, \| \, PW) + \Big|{A\mathstrut}_{TV}-R + \big|{B\mathstrut}_{TV} - K \big|^{+}\Big|^{+}\Big\}
\label{eqConv} \\
& \geq \;\;
\sup_{\substack{\\0\,\leq\,\rho\,\leq\,1}} \;
\nonumber \\
&
\min_{\substack{\\ TV}}\;
\Big\{D(TV \, \| \, PW) + \!\Big[{A\mathstrut}_{TV}-R + \big|{B\mathstrut}_{TV} - K \big|^{+}\Big]\cdot \rho\Big\}.
\label{eqLowerConvEnv}
\end{align}
Observe further, that
the lower bound (\ref{eqLowerConvEnv}) is the {\em lower convex envelope} of (\ref{eqConv}) as a function of $R \in \mathbb{R}$. Indeed,
the minimum (\ref{eqConv}) is a non-increasing function of $R$,
and therefore it cannot have
{\em lower supporting lines} 
with slopes greater than $0$.
It also cannot have lower supporting lines with negative slopes below $-1$,
as it decreases with the slope exactly $-1$ in the region of negative or small positive values of $R$.
Note that the objective function of the minimum (\ref{eqLowerConvEnv}) is continuous in $TV$ in the closed region of $TV \ll PW$.
Let $T_{\!\rho}\,V_{\!\rho}$ be the minimizing distribution of the minimum in (\ref{eqLowerConvEnv}) for a given $\rho \in [0, 1]$.
For this distribution there exists a real $R = R(\rho)$ such that the expression in the square brackets of (\ref{eqLowerConvEnv}) is zero.
Therefore, there is equality between (\ref{eqLowerConvEnv}) and (\ref{eqConv}) at $R(\rho)$. And this is achieved for each $\rho \in [0, 1]$,
which corresponds to lower supporting lines of slopes $-\rho$ between $0$ and $-1$.

Finally observe that there is in fact equality between (\ref{eqConv}) and (\ref{eqLowerConvEnv}) for all $R$,
since (\ref{eqConv}) is a convex ($\cup$) function of $R$ and therefore it coincides with its lower convex envelope.
Indeed, using the property $|\,t\,|^{+} = \max\,\{0, t\}$, the objective function of the minimization (\ref{eqConv}) can be rewritten as a maximum of three terms:
\begin{align}
\max
\Big\{& D(TV \, \| \, PW),
\nonumber \\
&
D(TV \, \| \, PW) + {A\mathstrut}_{TV} - R,
\nonumber \\
&
D(TV \, \| \, PW) + {A\mathstrut}_{TV} - R + {B\mathstrut}_{TV} - K
\Big\}.
\nonumber
\end{align}
Then, this objective function is convex ($\cup$) in the triple $(TV, R, K)$, as a maximum of convex ($\cup$) functions of $(TV, R, K)$.
In particular, the convexity of ${A\mathstrut}_{TV} = D(V \, \| \, P \, | \, T) \equiv D(TV \, \| \, TP)$ in $TV$
follows by the log-sum inequality \cite{CoverThomas}.
By the definition of convexity it is then verified that the minimum (\ref{eqConv}) itself is a convex ($\cup$) function of $R$.

So far we have shown that (\ref{eqConv}) and (\ref{eqLowerConvEnv}) are equal.
Consider now the minimum of (\ref{eqLowerConvEnv}) with any $\rho \in [0, 1]$:
\begin{align}
&
\min_{\substack{\\ TV}}\;
\Big\{D(TV \, \| \, PW) + \rho \Big[{A\mathstrut}_{TV}-R + \big|{B\mathstrut}_{TV} - K \big|^{+}\Big]\Big\}
\label{eqInnerMin} \\
& \geq \;\; \sup_{\substack{\\0\,\leq\,\beta\,\leq\,1}} \;
\nonumber \\
&
\min_{\substack{\\ TV}}\;
\Big\{D(TV \, \| \, PW) + \rho \Big[{A\mathstrut}_{TV}-R + \big[{B\mathstrut}_{TV} - K \big]\cdot \beta\Big]\Big\}.
\label{eqLowerConvEnv2}
\end{align}
By the same reasoning as before, there is equality also between (\ref{eqInnerMin}) and (\ref{eqLowerConvEnv2}). 
Putting together (\ref{eqLowerConvEnv2}) and (\ref{eqLowerConvEnv}) and denoting $\beta\cdot\rho = \eta$, we can rewrite (\ref{eqConv}) as
\begin{align}
& \sup_{\substack{\\0\,\leq\,\rho\,\leq\,1\\
0\,\leq\,\eta\,\leq\,\rho
}} \;
\min_{\substack{\\ TV}} \;
\Big\{D(TV \, \| \, PW) + \rho \big[{A\mathstrut}_{TV}-R\big] + \eta\big[{B\mathstrut}_{TV} - K \big]\Big\}
\nonumber \\
= \;\; & \sup_{\substack{\\0\,\leq\,\rho\,\leq\,1 \\
0\,\leq\,\eta\,\leq\,\rho
}} \;
\min_{\substack{\\ TV}}\;
\Big\{\underbrace{D(T \,\|\, T_{\!\rho,\,\eta}) + (1+\rho)D(V \,\|\, V_{\!\rho,\,\eta}\,|\,T)}_{\geq \; 0}
+
{E\mathstrut}_{0}(\rho, \eta, P)
- \rho R - \eta K\Big\},
\nonumber
\end{align}
where the minimizing solution is
\begin{equation} \label{eqMinimSol}
T_{\!\rho,\,\eta}(y)V_{\!\rho,\,\eta}(x\,|\,y) \; \equiv \; \frac{1}{c} \cdot P(x)W^{\frac{1+\eta}{1+\rho}}(y\,|\,x)\bigg[\sum_{\widetilde{x}}P(\widetilde{x})W^{\frac{1+\eta}{1+\rho}}(y\,|\,\widetilde{x})\bigg]^{\rho}.
\end{equation}
\end{proof}

\bigskip

\section{A converse theorem for the error and correct-decoding exponents}\label{Con}

\bigskip

Let ${P\mathstrut}_{\!e}({\cal C\mathstrut}_{\!n})$ denote the 
average error
probability of the maximum-likelihood decoder 
for a given codebook $\,{\cal C\mathstrut}_{\!n}$ of block length $n$, averaged over all messages and channel instances.
Let ${I\mathstrut}_{TV} = \min_{\,P}{A\mathstrut}_{TV}^{P}$ denote the mutual information of a pair of random variables with the joint distribution $TV$
and let us define:
\begin{align}
& {E\mathstrut}_{\!c}(R, K) \; \triangleq \;
\;\;\;\;\;\;\;\;\;\;\;\;\;\;\;\,\,\,\,
\min_{\substack{\\ PU}}\;
\;\;\;\;\;\;\;\;\,\,
\Big\{D(U\,\|\,W\,|\,P)
+ \big|R - {I\mathstrut}_{\!PU} - \big|{B\mathstrut}_{\!PU} - K \big|^{+}\big|^{+}\Big\},
\label{eqCD} \\
& {E\mathstrut}_{\!e}(R, K) \; \triangleq \;
\max_{\substack{\\ P}}\min_{\substack{\\ U:\\
{I\mathstrut}_{\!PU} \, + \, |{B\mathstrut}_{\!PU} \, - \, K |^{+} \; \leq \; R
}}
\Big\{D(U\,\|\,W\,|\,P)\Big\},
\label{eqSP}
\end{align}
where $P$ and $U$ 
are such that $\,U(y\,|\,x)P(x) \equiv T(y)V(x\,|\,y)$.
Then we can show the following

\bigskip

\begin{thm}[Converse bounds] \label{thm2}
\begin{align}
& \liminf_{n\,\rightarrow\,\infty} \;
\;\min_{\substack{\\ {\cal C\mathstrut}_{\!n}}}\,\,\,
-\tfrac{1}{n}\log
\big[1 - {P\mathstrut}_{\!e}({\cal C\mathstrut}_{\!n})\big] 
\;\; \geq \;\;
{E\mathstrut}_{\!c}(R, K),
\label{eqPcorrect} \\
& \limsup_{n\,\rightarrow\,\infty} \; \max_{\substack{\\ {\cal C\mathstrut}_{\!n}}}\;
-\tfrac{1}{n}\log
\;\;\;\;
{P\mathstrut}_{\!e} 
({\cal C\mathstrut}_{\!n})
\;\;\;\;\,\,
\;\; \leq \;\;
{E\mathstrut}_{\!e}(R, K),
\label{eqPerrorSP}
\end{align}
{\em where (\ref{eqPcorrect}) holds for all $(R, K)$ and (\ref{eqPerrorSP}) holds a.e.: except possibly for such $R(K)$ where there is
a transition
(a jump) from $+\infty$
to a finite value
of (\ref{eqSP})
as a
monotonically non-increasing
function of $R$.}
\end{thm}

\bigskip

Let ${P\mathstrut}_{\!e}({\cal C\mathstrut}_{\!n}\, | \, TV)$ denote
the {\em conditional} average error probability of
the maximum-likelihood decoder
for a codebook $\,{\cal C\mathstrut}_{\!n}$,
given that the joint type of the sent and the received blocks is $TV$.
Theorem~\ref{thm2} is 
a corollary of the following 
upper bound on the corresponding conditional probability of correct decoding:

\bigskip

\begin{lemma} \label{lemmacond}
{\em For any constant composition codebook $\,{\cal C\mathstrut}_{\!n}$ and any $\epsilon > 0$}
\begin{align}
&
1 - {P\mathstrut}_{\!e}({\cal C\mathstrut}_{\!n}\, | \, TV) \;\; \leq \;\; e^{-n\big(R \, - \, {I\mathstrut}_{TV} \, - \, {|{B\mathstrut}_{TV} \, - \, K|\mathstrut}^{+} \, - \, \epsilon \, + \, o(1)\big)},
\label{eqUpperCond}
\end{align}
{\em where the term $o(1)$, vanishing uniformly w.r.t. $TV$ as $n\rightarrow \infty$, depends on $\epsilon$, but does not depend on the choice of $\,{\cal C\mathstrut}_{\!n}$. 
}
\end{lemma}

\bigskip

\begin{proof}
Suppose we are given a constant composition codebook $\,{\cal C\mathstrut}_{\!n}$, where all $\lfloor e^{nR}\rfloor$ codewords are of the same type with empirical probabilities $P(x)$. 
Looking at the
codebook $\,{\cal C\mathstrut}_{\!n}$ as a matrix of letters from ${\cal X}$, of size $\lfloor e^{nR}\rfloor \times n$,
we construct a whole ensemble of block codes, by
permuting 
the columns of the matrix.
Observe that the total number of code permutations in the ensemble is given by
\begin{displaymath}
J \;\; \triangleq \;\; e^{n\big({H\mathstrut}_{\!P} \, + \, o(1)\big)} \; \cdot \; {\pi\mathstrut}_{\!P},
\end{displaymath}
where ${H\mathstrut}_{\!P}$ is the entropy of the empirical distribution $P$,
and ${\pi\mathstrut}_{\!P}$ denotes the number of same-symbol permutations in the type $P$,
i.e., the symbol permutations that do not change a codeword 
that is 
a member of
the type.

Suppose that for each code 
in the ensemble 
a separate independent channel instance is generated. 
And suppose that for every transmission 
one code in the ensemble (known to the decoder with its own channel instance) is chosen randomly 
with uniform probability over permutations.
Consider an event where 
the sent codeword, chosen with uniform probability over the code permutations and the messages, 
together with
the received vector 
have a joint type $TV$, such that $P(x) = \sum_{\,y}T(y)V(x\,|\,y)$.
Since the channel-generating distribution is memoryless,
this 
will result in the same conditional average probability of 
correct decoding given $TV$,
when averaged over all messages and channel instances,
as ${\cal C\mathstrut}_{\!n}$ itself.
In what follows, we will derive an upper bound on this probability. 

Let $\widehat{\bf y}$ be the received vector of the type $T$.
Consider the conditional type class ${\cal T}(V \, | \, \widehat{\bf y})$ of codewords with the 
empirical distribution $V$ given the vector $\widehat{\bf y}$.
Observe that the total number of all codewords in the ensemble belonging to this conditional type class 
(counted as distinct if corresponding to different code permutations or messages) is given by
\begin{equation} \label{eqS}
S \;\; \triangleq \;\; \underbrace{e^{n\big({H\mathstrut}_{\!V|T} \, + \, o(1)\big)} \; \cdot \; {\pi\mathstrut}_{\!P}\,}_{\text{for a message $m$}} \; \cdot \; \lfloor e^{nR}\rfloor,
\end{equation}
where ${H\mathstrut}_{\!V\,|\,T}$ is the average entropy of the conditional distribution $V$ given $T$,
i.e., ${H\mathstrut}_{\!V\,|\,T} = {\mathbb{E}\mathstrut}_{TV}\big[\!-\!\log V(X\,|\,Y)\big]$.

Let us fix two small numbers $\epsilon > \delta > 0$ and consider separately two cases.
Suppose first that $K - {B\mathstrut}_{TV} \geq \epsilon$.
In this case, the probability of an event, that
the cloud of any ${\bf x} \in {\cal T}(V \, | \, \widehat{\bf y})$ in the ensemble contains less than $e^{n(K \, - \,  {B\mathstrut}_{TV} \, - \, \delta)}$
or more than $e^{n(K \, - \,  {B\mathstrut}_{TV} \, + \, \delta)}$ vectors $\widehat{\bf y}$,
by Lemma~\ref{lemmaA} uniformly
tends to zero super-exponentially fast with $n$. Denote the complementary highly-probable event as $\Omega(\widehat{\bf y})$.
Let $k$ be an index of a code in the ensemble.
Let $N(k)$ denote the number of codewords from the conditional type class ${\cal T}(V \, | \, \widehat{\bf y})$ in the code of index $k$.
Then, given the conditions that the received vector is $\widehat{\bf y}$, that the sent codeword belongs to ${\cal T}(V \, | \, \widehat{\bf y})$, and $\Omega(\widehat{\bf y})$,
we have that the conditional probability of the code $k$ is upper-bounded by $N(k)e^{2n\delta}/S$.
Furthermore, given that indeed the code $k$ is used for communication,
the conditional probability of correct decoding is upper-bounded by $e^{2n\delta}/N(k)$.
Summing up over all codes, we can write
\begin{align}
\Pr\big\{\text{correct decoding}\,\big|\, TV, \,\widehat{\bf y},  \,\Omega(\widehat{\bf y})\big\} \;\;
&
\leq \;\;
\sum_{k:\; N(k)\, > \, 0}\frac{N(k)e^{2n\delta}}{S}\cdot\frac{e^{2n\delta}}{N(k)}
\;\;
\leq
\;\;
J
\cdot
\frac{e^{4n\delta}}{S} \;\;
\nonumber \\
& = \;\; e^{-n\big(R \, - \, {I\mathstrut}_{TV} \, - \, 4\delta \, + \, o(1)\big)}
\label{eqFirstCaseFirst} \\
& \leq \;\; e^{-n\big(R \, - \, {I\mathstrut}_{TV} \, - \, {|{B\mathstrut}_{TV} \, - \, K|\mathstrut}^{+} \, - \, 5\epsilon \, + \, o(1)\big)}.
\label{eqFirstCase}
\end{align}

Consider now the second case when $K - {B\mathstrut}_{TV} < \epsilon$.
In this case, the probability of an event, that
the cloud of any ${\bf x} \in {\cal T}(V \, | \, \widehat{\bf y})$ in the ensemble contains
more than $e^{n(\epsilon \, + \, \delta)}$ occurrences of the vector $\widehat{\bf y}$,
by (\ref{eqIncor}) of Lemma~\ref{lemmaA} uniformly
tends to zero super-exponentially fast.
Let us denote this rare event as ${\cal E\mathstrut}_{\!1}(\widehat{\bf y})$.
In fact, among the codewords ${\bf x} \in {\cal T}(V \, | \, \widehat{\bf y})$, 
those with clouds
containing $\widehat{\bf y}$ 
become rare.
However, the probability of an event, that in the ensemble
there are less than $S \cdot e^{-n\big({|{B\mathstrut}_{TV} \, - \, K|\mathstrut}^{+} \, + \, \epsilon\big)}$ codewords
from ${\cal T}(V \, | \, \widehat{\bf y})$
having at least one vector $\widehat{\bf y}$ in their cloud,
uniformly tends to zero super-exponentially fast.
This in turn can be verified similarly to (\ref{eqCor}) of Lemma~\ref{lemmaA}, using (\ref{eqS}).
Let us denote this rare event as ${\cal E\mathstrut}_{\!2}(\widehat{\bf y})$.
Let us denote the complementary (to 
the union of the events ${\cal E\mathstrut}_{\!1}(\widehat{\bf y})$ and ${\cal E\mathstrut}_{\!2}(\widehat{\bf y})$) highly-probable event as $\widetilde{\Omega}(\widehat{\bf y})$.

Let $\widetilde{N}(k)$ denote the number of such codewords in the code $k$, which {\em both} belong to the conditional type class ${\cal T}(V \, | \, \widehat{\bf y})$
{\em and} have at least one $\widehat{\bf y}$ in their respective cloud.
Then, given the 
intersection of events
that the received vector is $\widehat{\bf y}$, that the sent codeword belongs to ${\cal T}(V \, | \, \widehat{\bf y})$, and $\widetilde{\Omega}(\widehat{\bf y})$,
we obtain that the conditional probability of the code $k$ is upper-bounded by $\widetilde{N}(k)e^{n\big({|{B\mathstrut}_{TV} \, - \, K|\mathstrut}^{+} \, + \, 2\epsilon \, + \, \delta\big)}/S$.
Given that the code $k$ is used for communication,
the conditional probability of correct decoding is upper-bounded by $e^{n(\epsilon \, + \, \delta)}/\widetilde{N}(k)$.
Repeating the steps leading to (\ref{eqFirstCaseFirst}), we obtain (\ref{eqFirstCase}) once again.
\end{proof}

\bigskip

{\em Proof of Theorem~\ref{thm2}:}
First we verify the bound on the correct-decoding exponent (\ref{eqPcorrect}).
It is enough to consider constant composition codes,
because they can asymptotically achieve the same exponent of the correct-decoding probability
as the best block codes,
as is shown in the beginning of
\cite[Lemma~5]{DueckKorner79} using a suboptimal encoder-decoder pair.

Thus, let $\,{\cal C\mathstrut}_{\!n}$ be a constant composition codebook of a type $P$.
Consider an event where 
the sent codeword 
together with
the received vector 
have a joint type $PU$. 
The exponent of the probability of such event is given by $D(U \, \| \, W \, | \, P) + o(1)\,$.

Adding to this the lower bound on the exponent of the conditional probability of correct decoding given $PU$ of Lemma~\ref{lemmacond}
in the following form:
\begin{equation} \label{eqLBCond}
\big|R - {I\mathstrut}_{\!PU} - \big|{B\mathstrut}_{\!PU} - K \big|^{+} \big|^{+} - \epsilon + o(1),
\end{equation}
minimizing the resulting expression over all {\em distributions} $PU$, discarding $o(1)$, and taking $\epsilon \rightarrow 0$, 
we obtain (\ref{eqCD}).

Next we establish the bound on the error exponent (\ref{eqPerrorSP}). Here also it suffices to consider constant composition codebooks $\,{\cal C\mathstrut}_{\!n}$, because there is only a polynomial number of different types in a general codebook
of block length $n$.

Turning (\ref{eqUpperCond}) into a lower bound on ${P\mathstrut}_{\!e}({\cal C\mathstrut}_{\!n}\, | \, TV)$,
we can obtain the following upper bound on the error exponent of $\,{\cal C\mathstrut}_{\!n}$:
\begin{align}
& \max_{\substack{\\ P \,\in\,{{\cal P}\mathstrut}_{\!n}({\cal X})}}\;\min_{\substack{\\ U:\\ \\
PU  \, \in \, {{\cal P}\mathstrut}_{\!n}({\cal X}\,\times\,{\cal Y}),
\\
{I\mathstrut}_{\!PU} \, + \, |{B\mathstrut}_{\!PU} \, - \, K |^{+} \; \leq \; R \, - \, 2\epsilon
}}
\!\!\!\!\!\!\!\!
\Big\{D(U \, \| \, W \, | \, P)\Big\}
+ o(1)
 \label{eqSPTypes} \\
\leq \;\;
& \max_{\substack{\\  P \,\in\,{{\cal P}\mathstrut}_{\!n}({\cal X})}}\;\min_{\substack{\\ U:\\
{I\mathstrut}_{\!PU} \, + \, |{B\mathstrut}_{\!PU} \, - \, K |^{+} \; \leq \; R \, - \, 3\epsilon
}}
\!\!\!\!\!\!\!\!
\Big\{D(U \, \| \, W \, | \, P)\Big\}
+ o(1)
\label{eqDistW} \\
\leq \;\; &
\;\;\,
\max_{\substack{\\ P}}\;
\;\;\,
\min_{\substack{\\ U:\\
{I\mathstrut}_{\!PU} \, + \, |{B\mathstrut}_{\!PU} \, - \, K |^{+} \; \leq \; R \, - \, 3\epsilon
}}
\!\!\!\!\!\!\!\!
\Big\{D(U \, \| \, W \, | \, P)\Big\}
+ o(1).
\label{eqDistU}
\end{align}
Here (\ref{eqSPTypes}) follows directly from Lemma~\ref{lemmacond} and the fact that the exponent of $PU$
is $D(U \, \| \, W \, | \, P) + o(1)$.
In (\ref{eqDistW}) we extend the inner minimization from conditional types to arbitrary distributions $U$ with the help of an additional $\epsilon$ in the minimization condition.
In (\ref{eqDistU}) we extend the outer maximization to arbitrary distributions $P$,
and as a result the maximum cannot decrease.

In the limit of a large $n$ the vanishing term $o(1)$ in (\ref{eqDistU}) disappears and we are left with $\epsilon$. In order to replace $\epsilon > 0$ with zero,
observe that both the objective function and the expression in the minimization condition of (\ref{eqDistU}) are convex ($\cup$) functions of $U$.
It follows that the inner minimum of (\ref{eqDistU}) is a convex ($\cup$) function of $\epsilon \in \mathbb{R}$.
Therefore (\ref{eqDistU}) itself, as a maximum of convex functions of $\epsilon$, is convex ($\cup$) in $\epsilon \in \mathbb{R}$.
We conclude that by continuity of a convex function the maximum in (\ref{eqDistU}) tends to (\ref{eqSP})
as $\epsilon \rightarrow 0$, with a possible exception when (\ref{eqDistU}) jumps to $+\infty$ exactly at $\epsilon = 0$,
which corresponds to the jump to $+\infty$ of (\ref{eqSP}) as a convex ($\cup$) function of $R$ exactly at $R$.
$\square$

\bigskip

\section{Alternative representation of the converse bounds}\label{AltRep}

\bigskip

In this section we develop alternative expressions for the converse bounds of Theorem~\ref{thm2}.
Using the properties that ${I\mathstrut}_{TV} = \min_{\,P}{A\mathstrut}_{TV}^{P}$ and $|\,t\,|^{+} = \max\,\{0, t\}$,
the expression (\ref{eqCD}) for
the lower bound 
of Theorem~\ref{thm2} can be written also as
$\min_{\,P} \, {E\mathstrut}_{\!c}(P, R, K)$,
where
\begin{align}
& {E\mathstrut}_{\!c}(P, R, K) \;\; \triangleq \;\;
\min_{\substack{\\ TV}}\;
\Big\{
D(TV \, \| \, PW) + \big| R \, - {A\mathstrut}_{TV}^{P} - \big|{B\mathstrut}_{TV}^{W} - K \big|^{+}\big|^{+}\Big\},
\label{eqCDD}
\end{align}
and ${A\mathstrut}_{TV}^{P}$ and ${B\mathstrut}_{TV}^{W}$ 
are defined in (\ref{eqDefs}). 
An alternative expression for 
(\ref{eqCDD})
is given by

\bigskip

\begin{lemma}[Alternative representation --- correct-decoding exponent] \label{lemma12}
\begin{align}
{E\mathstrut}_{\!c}(P, R, K) \;\; = \;\; 
\min \bigg\{
&
\sup_{\substack{\\0 \, \leq \, \rho \, < \, 1}}
\Big\{{E\mathstrut}_{0}(-\rho, \,0, \,P) + \rho R\Big\},
\label{eqExplicit2} \\
&
\sup_{\substack{\\0 \, \leq \, \rho \, < \, 1}}
\Big\{{E\mathstrut}_{0}(-\rho, -\rho, \,P) + \rho (R + K)\Big\}
\bigg\},
\nonumber
\end{align}
{\em where ${E\mathstrut}_{\!c}(P, R, K)$ is defined by (\ref{eqCDD}) and  ${E\mathstrut}_{0}$ is defined as in (\ref{eqE0Def}).}
\end{lemma}

\bigskip

\begin{proof}
We can rewrite (\ref{eqCDD}) as a minimum of two terms:
\begin{align}
\min\bigg\{
&
\min_{\substack{\\ TV}}\;
\Big\{
D(TV \, \| \, PW) + \big| R \, - {A\mathstrut}_{TV} \big|^{+}\Big\},
\nonumber \\
&
\min_{\substack{\\ TV}}\;
\Big\{
D(TV \, \| \, PW) + \big| R \, - {A\mathstrut}_{TV} - {B\mathstrut}_{TV} + K \big|^{+}\Big\}
\bigg\}.
\nonumber
\end{align}
Solution of each one of the terms is similar to the method of Lemma~\ref{lemma11} and gives (\ref{eqExplicit2}).
\end{proof}

\bigskip

The expression (\ref{eqSP}) for the upper bound of Theorem~\ref{thm2} can be written alternatively as

\bigskip

\begin{lemma}[Alternative representation --- upper bound on the error exponent] \label{lemmaSP}
\begin{align}
& {E\mathstrut}_{\!e}(R, K) \; = \;
\max_{\substack{\\ P}}
\sup_{\substack{\\0\,\leq\,\eta \, \leq \, \rho}}
\Big\{{E\mathstrut}_{0}(\rho, \eta, P) - \rho R - \eta K\Big\},
\label{eqAlterSP}
\end{align}
{\em where ${E\mathstrut}_{\!e}(R, K)$ and ${E\mathstrut}_{0}$ are defined in (\ref{eqSP}) and (\ref{eqE0Def}), respectively.}
\end{lemma}

\bigskip

The proof is given in the Appendix. Examples of this bound together with the achievable error exponent as a lower bound are given in Fig.~\ref{graph3}.
Note the discontinuities (jumps to $+\infty$) in the upper bounds.
Observing the alternative to (\ref{eqAlterSP}) expression (\ref{eqStart}), which appears in the proof of Lemma~\ref{lemmaSP},
it can be verified similarly to Lemma~\ref{lemma11} that the discontinuity (jump to $+\infty$) in (\ref{eqAlterSP}) occurs
at
\begin{align}
{R\mathstrut}_{\min}(K) \; & = \; \max_{\substack{\\ P}} \;\;
\min_{\substack{\\ TV}}
\;\;
\Big\{D(V\,\|\,P\,|\,T) + |{B\mathstrut}_{TV} - K |^{+}\Big\}\;
\nonumber \\
& = \;
\max_{\substack{\\ P}} \sup_{\substack{\\ 0\,<\,\beta\,\leq\,1}}\bigg\{-\log \max_{y}\sum_{x}P(x){W\mathstrut}^{\beta}(y\,|\,x) - \beta K\bigg\}.
\nonumber
\end{align}
For $W = \text{BSC}(p)$ this gives $
{R\mathstrut}_{\min}(K) = \sup_{\,0\,<\,\beta\,\leq\,1}\big\{-\log\left[\tfrac{1}{2}(1-p)^{\beta} + \tfrac{1}{2}p^{\beta}\right]-\beta K\big\}$, so that 
there is no jump for $K \geq -\tfrac{1}{2}\log \big[p(1-p)\big]$.

\bigskip

\begin{figure}[!t]
\centering
  \centering
  \includegraphics[width=.48\textwidth]{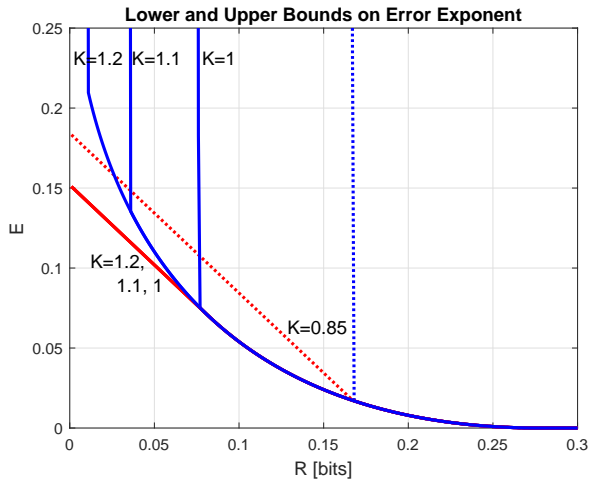}   
\caption{Lower and upper bounds on the optimal error exponent of the channel ensemble generated with $W = \text{BSC}(0.2)$,
as functions of $R$, for $K = 1.2, 1.1, 1,$ and $0.85$. The lower and the upper bounds were obtained by (\ref{eqExplicit1}) and (\ref{eqAlterSP}),
respectively, with $P = (0.5, 0.5)$.}
\label{graph3}
\end{figure}

\bigskip

\section{The capacity of the channel ensemble}\label{Cap}

\bigskip

Let us define the capacity of the channel ensemble generated with $W$, denoted as $C(W, K)$,
as the supremum of rates $R$, for which there exists a sequence of codebooks ${\cal C\mathstrut}_{\!n}$
of size $\lfloor e^{nR}\rfloor$ with ${P\mathstrut}_{\!e}({\cal C\mathstrut}_{\!n}) \rightarrow 0$
as $n\rightarrow \infty$.
Then, by the results of the previous sections, this corresponds to
the point on the $R$-axis, at which both the maximal achievable error exponent and the minimal correct-decoding exponent of the channel ensemble meet zero.
An example is shown in Fig.~\ref{graph12}.

\bigskip

\begin{thm}[Ensemble capacity] \label{thm3}
\begin{displaymath}
C(W, K) \;\; = \;\; \max\,\big\{C(W), \, {H\mathstrut}_{\!\max}(W) - K \big\},
\end{displaymath}
{\em where $C(W)$ is the Shannon capacity of the DMC
$W$, and ${H\mathstrut}_{\!\max}(W) \triangleq \max_{\, P}{H\mathstrut}_{\!T}\,$ with $\,T(y) = \sum_{\,x}P(x)W(y\,|\,x)$.}
\end{thm}

\bigskip

\begin{figure}[!t]
\centering
  \centering
  \includegraphics[width=.48\textwidth]{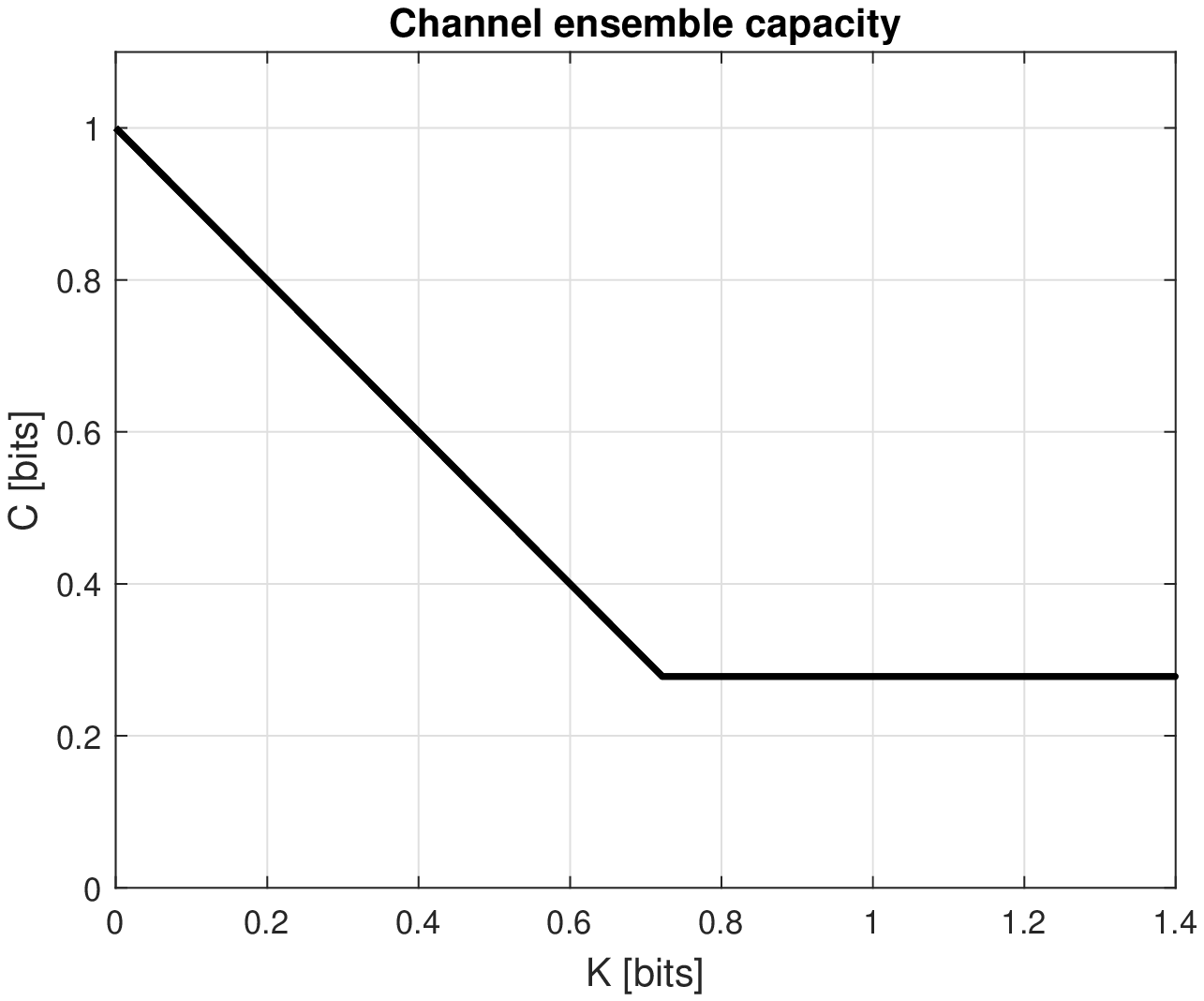}   
  \centering
  \includegraphics[width=.48\textwidth]{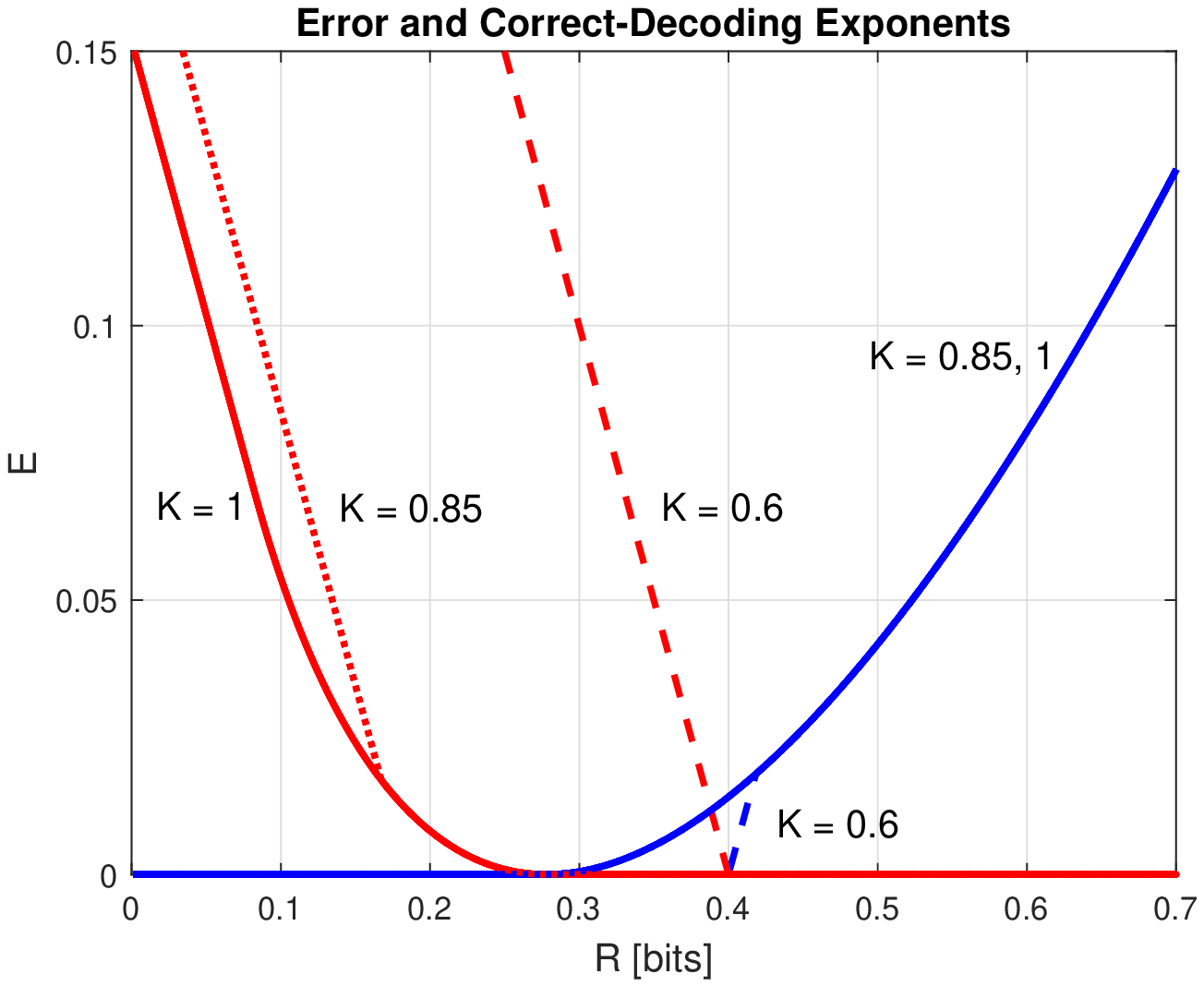}                                 
\caption{Left graph: The channel ensemble capacity $C(W, K)$ vs. $K$, with $W = \text{BSC}(0.2)$. Right graph:
Achievable error exponents (decreasing curves) and converse correct-decoding exponents (increasing curves) as functions of $R$, for $K = 1$, $0.85$, $0.6$, for the channel-generating distribution $W = \text{BSC}(0.2)$. The curves were obtained by (\ref{eqExplicit1}) and (\ref{eqExplicit2}) with $P = (0.5, 0.5)$. Note that (\ref{eqExplicit2}) for $K = 0.6$ is not convex in $R$.}
\label{graph12}
\end{figure}

\begin{proof}
The maximal achievable error exponent, provided by Theorem~\ref{thmAch}, is
\begin{displaymath}
\max_{\substack{\\P}} \; {E\mathstrut}_{\!e}(P, R, K),
\end{displaymath}
where ${E\mathstrut}_{\!e}(P, R, K)$ is given by (\ref{eqOneDef}).
The lower bound on the minimal correct-decoding exponent, given by Theorem~\ref{thm2}, can be written as
\begin{displaymath}
\min_{\substack{\\P}} \; {E\mathstrut}_{\!c}(P, R, K),
\end{displaymath}
where ${E\mathstrut}_{\!c}(P, R, K)$ is given by
(\ref{eqCDD}).
Since $D(TV \, \| \, PW)
 = 0$ iff $\,T(y)V(x\,|\,y) = P(x)W(y\,|\,x)$ for all $(x, y)$,
both expressions 
(\ref{eqOneDef})
and (\ref{eqCDD}), as functions of $R$, meet zero at the same point, which is $R = {A\mathstrut}_{PW}^{P} + \big|{B\mathstrut}_{PW}^{W} - K \big|^{+}$.
This gives
\begin{align}
C(W, K) \; & = \; \max_{\substack{\\P}}\,\Big\{{A\mathstrut}_{PW}^{P} + \big|{B\mathstrut}_{PW}^{W} - K \big|^{+}\Big\}
\nonumber \\
& = \; \max_{\substack{\\P}}\; \max \,\Big\{{A\mathstrut}_{PW}^{P}, \;  {A\mathstrut}_{PW}^{P} + {B\mathstrut}_{PW}^{W} - K \Big\}
\nonumber \\
& = \;
\max \,\Big\{\max_{\substack{\\P}}{A\mathstrut}_{PW}^{P}, \;  \max_{\substack{\\P}} \big\{{A\mathstrut}_{PW}^{P} + {B\mathstrut}_{PW}^{W} - K\big\} \Big\}
\nonumber \\
& = \;
\max \,\Big\{C(W), \;  {H\mathstrut}_{\!\max}(W) - K \Big\},
\nonumber
\end{align}
where the last equality follows because ${A\mathstrut}_{PW}^{P} + {B\mathstrut}_{PW}^{W} = {H\mathstrut}_{\!T}$ with the corresponding $T(y) = \sum_{\,x}P(x)W(y\,|\,x)$, $\forall y$.
\end{proof}

\bigskip

\section{The optimal correct-decoding exponent}\label{CorDec}

\bigskip

In fact, the lower bound (\ref{eqPcorrect}) is achievable.
As in Section~\ref{Ach}, suppose the codebook is generated i.i.d. according to a distribution over ${\cal X}$ with probabilities $P(x)$,
and let $\,\overline{\!{P\mathstrut}}{\mathstrut}_{\!e}^{\,(n)}$ denote the average error probability 
of the maximum-likelihood decoder, averaged over all possible messages, codebooks, and channel instances.

\bigskip

\begin{lemma}[Achievable correct-decoding exponent] \label{lemmaUpCorDec}
\begin{equation} \label{eqPcorrUp}
\limsup_{n\,\rightarrow\,\infty} \,-\tfrac{1}{n}\log 
\big[1 - \,\overline{\!{P\mathstrut}}{\mathstrut}_{\!e}^{\,(n)}\big]
\;\; \leq \;\;
{E\mathstrut}_{\!c}(P, R, K),
\end{equation}
{\em where ${E\mathstrut}_{\!c}(P, R, K)$ is defined in (\ref{eqCDD}).}
\end{lemma}

\bigskip

\begin{proof}
Consider the following suboptimal decoder.
The decoder works with a single anticipated joint type $TV$ of the sent codeword ${\bf x}$ and the received vector $\widehat{\bf y}$.
If the type of $\widehat{\bf y}$ is not $T$, the decoder declares an error.
Otherwise, in case the type of the received block is indeed $T$, the decoder looks for the indices of the codewords
with the conditional type $V$ w.r.t. $\widehat{\bf y}$, with {\em at least one} replica of $\widehat{\bf y}$ in their clouds,
and chooses one of these indices as its estimate $\widehat{m}$ of the transmitted message. The choice is made randomly with uniform probability, regardless of the
actual number of replicas of $\widehat{\bf y}$ in each cloud.
If there are no codewords of the conditional type $V$ w.r.t. $\widehat{\bf y}$ with at least one $\widehat{\bf y}$ in their cloud,
then again the decoder declares an error.

Let ${N\mathstrut}_{\!nc}$ denote the random number of {\em incorrect} codewords of the conditional type $V$ w.r.t. $\widehat{\bf y}$, with at least one replica of $\widehat{\bf y}$ in their clouds, in the codebook. Then the conditional probability of the correct decoding, given that the joint type of the received and the transmitted blocks is indeed $TV$,
is given by
\begin{equation} \label{eqJensenN}
\mathbb{E}\!\left[\frac{1}{{N\mathstrut}_{\!nc}+1}\right] \;\; 
\geq
\;\;
\frac{1}{\mathbb{E}[{N\mathstrut}_{\!nc}\,] + 1},
\end{equation}
with Jensen's inequality
where the expectation is w.r.t. the randomness of both the incorrect codewords and their clouds.
Note that the exponent of $\mathbb{E}[{N\mathstrut}_{\!nc}\,]$ can be expressed as $R$ minus (\ref{eqExponentS}) with $TV$.
The RHS of (\ref{eqJensenN}) then results in the following upper bound on the exponent of the conditional probability of correct decoding:
\begin{equation} \label{eqCorCondExp}
\big|R - {A\mathstrut}_{TV} - \big|{B\mathstrut}_{TV} - K \big|^{+} \big|^{+} + o(1).
\end{equation}
Adding to this the exponent of the joint type $TV$, we obtain (\ref{eqCDD}).
\end{proof}

\bigskip

Now, since ${E\mathstrut}_{\!c}(R, K) = \min_{\,P} \, {E\mathstrut}_{\!c}(P, R, K)$,
by (\ref{eqPcorrect}) of Theorem~\ref{thm2} and Lemma~\ref{lemmaUpCorDec} we have the following

\bigskip

\begin{thm}[Optimal correct-decoding exponent] \label{thmexact}
\begin{align}
& \lim_{n\,\rightarrow\,\infty} \;
\;\min_{\substack{\\ {\cal C\mathstrut}_{\!n}}}\,\,\,
-\tfrac{1}{n}\log
\big[1 - {P\mathstrut}_{\!e}({\cal C\mathstrut}_{\!n})\big] 
\;\; = \;\;
{E\mathstrut}_{\!c}(R, K),
\label{eqExact}
\end{align}
{\em where ${E\mathstrut}_{\!c}(R, K)$ is defined in (\ref{eqCD}) and ${P\mathstrut}_{\!e}({\cal C\mathstrut}_{\!n})$ is defined as in Section~\ref{Con}.
This exponent is achievable by random coding.}
\end{thm}

\newpage

\section*{Appendix}\label{Appendix}

\bigskip

{\em Proof of Lemma~\ref{lemmaB}:}
The LHS of (\ref{eqConstBound}) equals $1$ for $\alpha = 1$. Suppose $\alpha \in (0, q\,]$.
\begin{align}
\Pr\left\{N \, \geq \, {N\mathstrut}_{\!1} + 1 \, | \, N \, \geq \, 1\right\} \; & = \;
\frac{\Pr\left\{N \, \geq \, {N\mathstrut}_{\!1} + 1, \, N \, \geq \, 1\right\}}{\Pr\left\{N \, \geq \,  1\right\}}
\nonumber \\
& = \;
\frac{\Pr\left\{N \, \geq \, {N\mathstrut}_{\!1} + 1\right\}}{\Pr\left\{N \, \geq \,  1\right\}}
\nonumber \\
& \overset{a}{\geq} \;
\frac{\Pr\left\{N \, \geq \, {N\mathstrut}_{\!1} + Z + 1\right\}}{\Pr\left\{N \, \geq \,  1\right\}}
\nonumber \\
& \overset{b}{=} \;
\frac{\Pr\left\{N \, \geq \, \widetilde{N} + 1\right\}}{\Pr\left\{N \, \geq \,  1\right\}}
\nonumber \\
& \overset{c}{=} \;
\frac{\Pr\left\{N \, \geq \, \widetilde{N} + 1\right\} + \Pr\left\{\widetilde{N} \, \geq \, N + 1\right\}}{2\Pr\left\{N \, \geq \,  1\right\}}
\nonumber \\
& = \;
\frac{\Pr\left\{N \, \neq \, \widetilde{N}\right\}}{2\Pr\left\{N \, \geq \,  1\right\}}
\nonumber \\
& \overset{d}{=} \;
\frac{1}{2}\cdot
\frac{1 - {p\mathstrut}^{2}(0) - \sum_{k\,=\,1}^{M}{p\mathstrut}^{2}(k)}{1 - p(0)}
\nonumber \\
& = \;
\frac{1}{2}\cdot
\left[1 + p(0) -
\sum_{k\,=\,1}^{M}\frac{p(k)}{1 - p(0)}\cdot p(k)
\right]
\nonumber \\
& \geq \;
\frac{1}{2}\cdot
\Big[1 + p(0) -
\max_{\substack{\\ 1\,\leq\,k\,\leq\,M}} p(k)
\Big]
\nonumber \\
& \geq \;
\frac{1}{2}\cdot
\Big[1 -
\max_{\substack{\\ 1\,\leq\,k\,\leq\,M}} p(k)
\Big]
\nonumber \\
& \overset{e}{\geq} \;
\frac{1}{2}\cdot\!
\Bigg[1 -
\,\,\,\,
\max\left\{
\frac{{e\mathstrut}^{\frac{1}{12M}}}{\sqrt{2\pi}}\sqrt{\frac{M}{M-1}}, \; {q\mathstrut}^{M}\right\}\,
\Bigg],
\nonumber
\end{align}
where in\newline
($a$) we add a nonnegative random variable $Z \sim \text{Bernoulli}(\alpha)$, statistically independent with $N$ and ${N\mathstrut}_{\!1}$;\newline
($b$) random variable $\widetilde{N} = {N\mathstrut}_{\!1} + Z \,\sim\, {\rm B}(M, \alpha)$ is statistically independent with $N$;\newline
($c$) we use the symmetry (i.i.d.) between $\widetilde{N}$ and $N$;\newline
($d$) we denote $p(k) \triangleq \Pr\left\{N = k\right\}$ and use the independence of $\widetilde{N}$ and $N$;\newline
($e$) for $k = M$ we use $p(M) \leq {q\mathstrut}^{M}$ and for $1 \leq k \leq M-1$ use Stirling's bounds 
\cite[Ch.~II, Eq.~9.15]{Feller}
to obtain
\begin{align}
p(k) \; = \; \binom{M}{k} {\alpha\mathstrut}^{k} {(1-\alpha)\mathstrut}^{M-\,k}
\; & \leq \;
\frac{{e\mathstrut}^{\frac{1}{12M}}}{\sqrt{2\pi}}\sqrt{\frac{M}{k(M-k)}}
{e\mathstrut}^{-M \cdot D(k/M\,\|\,\alpha)}
\nonumber \\
& \leq \;
\frac{{e\mathstrut}^{\frac{1}{12M}}}{\sqrt{2\pi}}\sqrt{\frac{M}{k(M-k)}}
\; \leq \;
\frac{{e\mathstrut}^{\frac{1}{12M}}}{\sqrt{2\pi}}\sqrt{\frac{M}{M-1}},
\label{eqBinomialUpperBound}
\end{align}
where $D(\cdot \, \|\, \cdot)$ is the binary Kullback-Leibler divergence.
$\square$

\bigskip

{\em Proof of Lemma~\ref{lemmaSP}:}
The proof can be done by a ``sandwich'' of the following inequalities: 
\begin{align}
&
\;\,
\max_{\substack{\\ P}}\min_{\substack{\\ \widetilde{P}U:\\
{A\mathstrut}_{\widetilde{P}U}^{P} \, + \, |{B\mathstrut}_{\!\widetilde{P}U} \, - \, K |^{+} \; \leq \; R
}}
\Big\{D(\widetilde{P}U \, \| \, PW)\Big\}
\label{eqStart} \\
\overset{a}{\leq} \;\; &
\;\,
\max_{\substack{\\ P}}\min_{\substack{\\ U:\\
{A\mathstrut}_{PU}^{P} \, + \, |{B\mathstrut}_{\!PU} \, - \, K |^{+} \; \leq \; R
}}
\Big\{D(PU \, \| \, PW)\Big\}
\label{eqConvR} \\
\overset{b}{=} \;\; &
\;\,
\max_{\substack{\\ P}}
\sup_{\substack{\\0\,\leq\,\eta \, \leq \, \rho}}
\,
\min_{\substack{\\ U}}
\;\;\;\;\;\;\;\;\;\;\,
\Big\{D(PU \, \| \, PW)+ \rho{A\mathstrut}_{PU}^{P} + \eta {B\mathstrut}_{\!PU} - \rho R - \eta K\Big\}
\label{eqEnvR} \\
= \;\; &
\sup_{\substack{\\0\,\leq\,\eta \, \leq \, \rho}}
\max_{\substack{\\ P}}
\;\,\,
\min_{\substack{\\ U}}
\;\;\;\;\;\;\;\;\;\;\,
\Big\{D(PU \, \| \, PW)+ \rho{A\mathstrut}_{PU}^{P} + \eta {B\mathstrut}_{\!PU} - \rho R - \eta K\Big\}
\nonumber \\
\overset{c}{\leq} \;\; &
\sup_{\substack{\\0\,\leq\,\eta \, \leq \, \rho}}
\max_{\substack{\\ P}}
\;\,\,
\;\;\;\;\;\;\;\;\;\;\;\;\;\;\;\;\;
\Big\{D(P{\widehat{U}\mathstrut}_{\!\!\rho, \, \eta} \, \| \, PW)+ \rho{A\mathstrut}_{P{\widehat{U}\mathstrut}_{\!\!\rho, \eta}}^{P} + \eta {B\mathstrut}_{\!P{\widehat{U}\mathstrut}_{\!\!\rho, \eta}} - \rho R - \eta K\Big\}
\nonumber \\
\overset{d}{=} \;\; &
\sup_{\substack{\\0\,\leq\,\eta \, \leq \, \rho}}
\max_{\substack{\\ P}}
\;\,\,
\;\;\;\;\;\;\;\;\;\;\;\;\;\;\;\;\;
\Big\{\rho\,{\mathbb{E}\mathstrut}_{P{\widehat{U}\mathstrut}_{\!\!\rho, \eta}}\big[\log \hat{c}\,\big]
- (1+\rho){\mathbb{E}\mathstrut}_{P{\widehat{U}\mathstrut}_{\!\!\rho, \eta}}\big[\log \hat{c}(X)\big]
-
\rho D\big(T\,\|\,{\widehat{T}\mathstrut}_{\!\!\rho, \,\eta}\big)
- \rho R - \eta K\Big\}
\nonumber \\
\overset{e}{=} \;\; &
\sup_{\substack{\\0\,\leq\,\eta \, \leq \, \rho}}
\;\;\;\;\;\;\;\;\;\;\;\;\;\;\;\;\;
\;\;\;\;\;\;\;\,\,\,\,
\Big\{{E\mathstrut}_{0}\big(\rho, \eta, {\widehat{P}\mathstrut}_{\!\!\rho, \, \eta}\big)
- \rho R - \eta K\Big\}
\nonumber \\
= \;\; &
\sup_{\substack{\\0\,\leq\,\eta \, \leq \, \rho}}
\max_{\substack{\\ P}}
\;\,\,
\;\;\;\;\;\;\;\;\;\;\;\;\;\;\;\;\;
\Big\{{E\mathstrut}_{0}(\rho, \eta, P)
- \rho R - \eta K\Big\}
\nonumber \\
= \;\; &
\;\,
\max_{\substack{\\ P}}
\sup_{\substack{\\0\,\leq\,\eta \, \leq \, \rho}}
\;\,
\;\;\;\;\;\;\;\;\;\;\;\;\;\;\;\;
\Big\{{E\mathstrut}_{0}(\rho, \eta, P)
- \rho R - \eta K\Big\},
\label{eqFinish}
\end{align}
where the transitions are as follows.\newline
($a$) The minimization over $\widetilde{P}$ is removed
by choosing $\widetilde{P}(x) \equiv P(x)$.\newline
($b$) The equivalence between the minimum in (\ref{eqConvR}) and the supremum in (\ref{eqEnvR})
can be shown by the method used in the proof of Lemma~\ref{lemma11}.\newline
($c$) The minimization over $U$ is removed by choosing $U(y\,|\,x)\equiv {\widehat{U}\mathstrut}_{\!\!\rho, \, \eta}(y\,|\,x)$
for each pair $(\rho, \eta)$,
where ${\widehat{U}\mathstrut}_{\!\!\rho, \, \eta}$ is the 
conditional distribution derived and extended to all $x \in {\cal X}$ from the joint distribution (\ref{eqMinimSol}):
\begin{equation} \label{eqUStar}
{\widehat{U}\mathstrut}_{\!\!\rho, \, \eta}(y\,|\,x) \; \equiv \; \frac{1}{\hat{c}(x)} \cdot W^{\frac{1+\eta}{1+\rho}}(y\,|\,x)\bigg[\sum_{\widetilde{x}}{\widehat{P}\mathstrut}_{\!\!\rho, \, \eta}(\widetilde{x})W^{\frac{1+\eta}{1+\rho}}(y\,|\,\widetilde{x})\bigg]^{\rho},
\end{equation}
where $\hat{c}(x)$ is the normalizing constant for each $x$ and ${\widehat{P}\mathstrut}_{\!\!\rho, \, \eta}$
is a special input distribution achieving the maximum of (\ref{eqE0Def}):
\begin{displaymath}
{\widehat{P}\mathstrut}_{\!\!\rho, \, \eta} \;\; \in \;\; \underset{P}{\arg\max}\;{E\mathstrut}_{0}(\rho, \eta, P).
\end{displaymath}
Note that the conditional distribution ${\widehat{U}\mathstrut}_{\!\!\rho, \, \eta}$ is well defined (adds up to $1$) for {\em all} $x \in {\cal X}$,
i.e., $\hat{c}(x) > 0, \forall x$.
This is because the normalizing constant $\hat{c}(x)$ is proportional to the derivative of the expression inside $\log$ of (\ref{eqE0Def})
w.r.t. $P(x)$ at ${\widehat{P}\mathstrut}_{\!\!\rho, \, \eta}(x)$,
and therefore, as in \cite[Eq.~22-23]{Arimoto76}, must {\em necessarily} satisfy
\begin{align}
\hat{c}(x) \;\; & = \;\; \hat{c}, \;\;\;\;\;\; {\widehat{P}\mathstrut}_{\!\!\rho, \, \eta}(x) > 0,
\nonumber \\
\hat{c}(x) \;\; & \geq \;\; \hat{c}, \;\;\;\;\;\; {\widehat{P}\mathstrut}_{\!\!\rho, \, \eta}(x) = 0.
\nonumber
\end{align}
Otherwise it would be possible to redistribute the probability mass ${\widehat{P}\mathstrut}_{\!\!\rho, \, \eta}(x)$ to further increase (\ref{eqE0Def}).\newline
($d$) 
The objective function is rewritten with
$\hat{c} = \sum_{x}{\widehat{P}\mathstrut}_{\!\!\rho, \, \eta}(x)\hat{c}(x)$,
$T(y) \equiv \sum_{x}P(x){\widehat{U}\mathstrut}_{\!\!\rho, \, \eta}(y\,|\,x)$,
and ${\widehat{T}\mathstrut}_{\!\!\rho, \,\eta}$ $\gg {\widehat{U}\mathstrut}_{\!\!\rho, \, \eta}(\,\cdot\,|\,x)$, $\forall x \in {\cal X}$,
which is the marginal distribution of $y$ derived from the joint distribution (\ref{eqMinimSol}) with ${\widehat{P}\mathstrut}_{\!\!\rho, \, \eta}\,$:
\begin{equation} \label{eqTStar}
{\widehat{T}\mathstrut}_{\!\!\rho, \,\eta}(y) \; \equiv \; \frac{1}{\hat{c}} \cdot
\bigg[\sum_{x}{\widehat{P}\mathstrut}_{\!\!\rho, \, \eta}(x)W^{\frac{1+\eta}{1+\rho}}(y\,|\,x)\bigg]^{1\,+\,\rho}.
\end{equation}
($e$) Observe that any distribution $P$ maximizes ${\mathbb{E}\mathstrut}_{P{\widehat{U}\mathstrut}_{\!\!\rho, \eta}}\big[\log \hat{c}\,\big]$,
while any distribution $P \ll {\widehat{P}\mathstrut}_{\!\!\rho, \, \eta}$ maximizes
${\mathbb{E}\mathstrut}_{P{\widehat{U}\mathstrut}_{\!\!\rho, \eta}}\big[-\log \hat{c}(X)\big]$.
Ultimately, the choice $P = {\widehat{P}\mathstrut}_{\!\!\rho, \, \eta}$
gives $D\big(T\,\|\,{\widehat{T}\mathstrut}_{\!\!\rho, \,\eta}\big) = 0$,
maximizing the whole expression, where ${E\mathstrut}_{0}\big(\rho, \eta, {\widehat{P}\mathstrut}_{\!\!\rho, \, \eta}\big) = -\log \hat{c}$.

Finally, observe that (\ref{eqConvR}) is the same as (\ref{eqSP}). The equivalence of the minimum in (\ref{eqStart}) and the supremum in (\ref{eqFinish})
can be shown as in the proof of Lemma~\ref{lemma11}. $\square$


\bibliographystyle{IEEEtran}

\end{document}